\documentclass[journal]{IEEEtran}
\usepackage{hyperref}       % hyperlinks
\usepackage{url}            % simple URL typesetting
\usepackage{booktabs}       % professional-quality tables
\usepackage{amsfonts}       % blackboard math symbols
\usepackage{nicefrac}       % compact symbols for 1/2, etc.
\usepackage{microtype}      % microtypography
\usepackage{xcolor}         % colors
\usepackage{cite}             % improves presentation of citations
\usepackage{array}
\usepackage{multirow}
\usepackage{lscape}
\usepackage[cmex10]{amsmath}  % Use the [cmex10] option to ensure compliance
\usepackage{bm,graphicx,tikz,subcaption,algorithm,algorithmic}
\usepackage{amsthm}
\usepackage{amsmath,amssymb}
\newcommand{\norm}[1]{%
	\| #1 \|}
\newcommand{\sigmax}{\sigma_{1}^*{}}
\newcommand{\sigmin}{\sigma_{r}^*{}}
\newcommand{\SE}{\mathrm{SD}}
\newcommand{\SEF}{\SE}
\newcommand{\x}{\bm{x}}
\newcommand{\y}{\bm{y}}
\newcommand{\g}{\bm{g}}
\renewcommand{\b}{\bm{b}}

\newcommand{\X}{\bm{X}}

\newcommand{\A}{\bm{A}}
\newcommand{\U}{{\bm{U}}}
\newcommand{\V}{\bm{V}}
\newcommand{\B}{\bm{B}}
\newcommand{\xstar}{\x^*}
\newcommand{\bstar}{\b^*}

\newcommand{\Xstar}{\X^*}
\newcommand{\Ustar}{\U^*}
\newcommand{\Vstar}{\V^*}
\newcommand{\Bstar}{\B^*}
\newcommand{\bSigma}{\bm{\Sigma}^*}

\newtheorem{theorem}{Theorem}
\newtheorem{cor}[theorem]{Corollary}

\newtheorem{assu}{Assumption}
\newtheorem{lemma}[theorem]{Lemma}
\newtheorem{fact}[theorem]{Fact}

\newcommand{\n}{\bm{v}}

\renewcommand{\S}{\bm{S}}

\newcommand{\ik}{{ki}}

\newcommand{\Xhat}{\hat\X}

\newcommand{\tC}{\tilde{C}}
\newcommand{\tc}{\tilde{c}}
\newcommand{\Q}{\bm{Q}}
\newcommand{\svdeq}{\overset{\mathrm{SVD}}=}
\newcommand{\sigmamin}{\sigma_{min}}

\newcommand{\ben}{\begin{enumerate}} \newcommand{\een}{\end{enumerate}} \newcommand{\bi}{\begin{itemize}} \newcommand{\ei}{\end{itemize}}

\newcommand{\indic}{\mathbf{1}}
\newcommand{\E}{\mathbb{E}}
\newcommand{\SD}{\bm{SD}}
\newcommand{\eps}{\epsilon}

\newcommand{\I}{\bm{I}}
\newcommand{\M}{\bm{M}}

\renewcommand{\a}{\bm{a}}
\newcommand{\e}{\bm{e}}

\newcommand{\G}{\bm{G}}

\newcommand{\Err}{\mathrm{Err}}

\newcommand{\Uhat}{\hat\U}

\newcommand{\D}{\bm{D}}

\newcommand{\Sigmastar}{\bm\Sigma^*}

\newcommand{\z}{\bm{z}}

\renewcommand\thetheorem{\arabic{section}.\arabic{theorem}}

\newcommand{\trnc}{\mathrm{trunc}}

\renewcommand{\Xhat}{\X}

\newcommand{\Bcheck}{\check{\V}}%{\check{\B}}
\newcommand{\w}{\bm{w}}

\newcommand{\qreq}{\overset{\mathrm{QR}}=} 
\newcommand{\bzeta}{\bm{\zeta}}
\newcommand{\s}{\bm{s}}
\begin{document}	
\title{Noisy Low Rank Column-wise Sensing}
\author{%
	\IEEEauthorblockN{Ankit Pratap Singh and Namrata Vaswani}\\
	\IEEEauthorblockA{Iowa State University, Ames, IA, USA}
}

\maketitle	
	
	\begin{abstract}
	   This letter studies the AltGDmin algorithm for solving the noisy low rank column-wise sensing (LRCS) problem. Our sample complexity guarantee improves upon the best existing one by a factor $\max(r,\log(1/\eps)) /r$ where $r$ is the rank of the unknown matrix and $\eps$ is the final desired accuracy.
A second contribution of this work is a detailed comparison of guarantees from all work that studies the exact same mathematical problem as LRCS, but refers to it by different names.
	\end{abstract}
\vspace{-0.2in}
\section{Introduction}
This work studies the Alternating Gradient Descent Minimization (AltGDmin) algorithm \cite{lrpr_gdmin} for solving the  noisy low rank column-wise sensing (LRCS) problem defined next.
LRCS finds important applications in federated sketching \cite{lee2019neurips}, accelerated dynamic MRI \cite{lrpr_gdmin_mri_jp}, and multi-task linear representation learning \cite{netrapalli}.
\vspace{-0.1in}
\subsection{Noisy LRCS Problem Setting}\label{problem}
The goal is to recover an $n \times q$ rank-$r$ matrix $\Xstar =[\xstar_1, \xstar_2, \dots, \xstar_q]$ from  $m$ noisy linear projections (sketches) of each of its $q$ columns, i.e. from
\begin{equation}\label{eq:prob}
	\y_k := \A_k \xstar_k + \n_k, \ k  \in [q]
\end{equation}
where each $\y_k$ is an $m$-length vector,  $[q]:=\{1,2,\dots, q\}$, the measurement/sketching matrices $\A_k$ are mutually independent and known, and noise $\n_k$ is an $m$-length vector which is independent of $\A_k$ and each entry of it is independent and identically distributed (i.i.d.) zero mean Gaussian with variance $\sigma_v^2$.
Each $\A_k$ is random-Gaussian: each entry of it is i.i.d. standard Gaussian. 
We denote the reduced Singular Value Decomposition (SVD) of the rank-$r$ matrix $\Xstar$ as
\[
\Xstar\svdeq\Ustar\Sigmastar\Vstar{}^\top=\Ustar\Bstar,
\]
and we let $\sigmax:=\sigma_{\max}(\Sigmastar)$ and $\sigmin:=\sigma_{\min}(\Sigmastar)$.
LRCS uses the following assumption \cite{lrpr_it,lrpr_best,lrpr_gdmin,lrpr_gdmin_2,netrapalli}.
\begin{assu}[Right Singular Vectors’ Incoherence]\label{right_incoh}
	Assume that right singular vectors' incoherence holds, i.e., $\max_{k \in [q]} \| \bstar_k \| \leq \mu \sqrt{r/q} \sigmax$ for a constant $\mu \geq 1$. %\sigma_1(\Xstar)
\end{assu}

\newcommand{\N}{\bm{V}}

We state guarantees in terms of the noise-to-signal ratio (NSR). NSR for LR recovery is the ratio between the maximum eigenvalue of $\E[\N \N^\top] = \sum_k \E[\n_k \n_k^\top] = q \sigma_v^2 \I$  to the minimum nonzero eigenvalue of $\Xstar \Xstar{}^\top$. This definition ensures that we are considering the ratio between the worst-case (largest) noise power in any direction to the smallest signal power in any direction. Since our ``signal'' matrix $\Xstar$ lies in an $r$-dimensional subspace of $\Re^n$,  the minimum nonzero eigenvalue of $\Xstar \Xstar{}^\top$ is its $r$-th eigenvalue which equals $\sigmin^2$. 
Noise power in any direction equals $q \sigma_v^2$. Thus,
\[
\mathrm{NSR}: =  \frac{q \sigma_v^2}{\sigmin{}^2}
\]

\vspace{-0.1in}
\subsubsection{Notation}\label{notat} We use $\|.\|_F$ to denote the Frobenius norm and $\|.\|$ without a subscript to denote the (induced) $l_2$ norm;  $^\top$ denotes matrix or vector transpose. We use $\e_k$ to denote the $k$-th canonical basis vector ($k$-th column of identity matrix $\I$); $|\bm{z}|$ for a vector $\bm{z}$ denotes element-wise absolute values;
and $\M^\dag = (\M^\top \M)^{-1} \M^\top$. We use $\indic_{\{a\leq b\}}$ to denote the indicator function that returns $1$ if $a\leq b$ otherwise $0$.
For tall matrices with orthonormal columns $\U_1, \U_2$, we use $\SD_2(\U_1, \U_2): = \|(\I - \U_1 \U_1{}^\top)\U_2\|$ as the Subspace Distance (SD) measure between the column spans of the two matrices. 
{\em We reuse the letters $c,C$ to denote different numerical constants in each use with the convention that $c < 1$ and $C \ge 1$.} {\em Whenever we use `with high probability (whp)' we mean with probability at least $1 - n^{-10}$.}

\vspace{-0.1in}
\subsection{Contribution}
In this work, we obtain guarantees for AltGDmin \cite{lrpr_gdmin} for solving the LRCS problem in a noisy setting. Our sample complexity guarantee improves upon the best existing one for the noisy case from \cite{netrapalli} by a factor $ \frac{\max(r,\log(1/\eps)) \log n}{r}$. Here $\eps$ is the final desired accuracy. It does this by extending the noise-free case result of \cite{lrpr_gdmin_2} for the noisy case. 
A second contribution of this work is a detailed comparison of guarantees from all work that studies the exact same mathematical problem as LRCS, but has been variously referred to as ``multi-task linear representation learning'', ``federated sketching of low rank matrices'' and ``low rank phase retrieval (LRPR)'' (which is its strict generalization).

\begin{table*}[h!]
	\centering
%	\resizebox{\linewidth}{!}{
	\begin{tabular}{|c|c|c|c|}
		\hline
		\textbf{Methods}                            & \textbf{Sample Comp. $mq\gtrsim$}   & \textbf{Time Complexity} & \textbf{Assumption} \\ \hline 
		AltMin \cite{lrpr_icml,lrpr_it}  & $nr^4 \log (\frac{1}{\eps})$ & $mnq\cdot r\log^2(\frac{1}{\eps})$ & RSV \& $\sigma_v = 0$ (only noise-free case) \\ \hline
	Convex \cite{lee2019neurips}                      & $\frac{nr}{\eps^4} \log^6 n \max\left( 1 ,\frac{\mathrm{NSR}}{r} \right)$
		                      & $mnq\cdot \min\left(\frac{1}{\sqrt{\eps}},n^3r^3\right)$    & RSV   \& $\norm{\Xstar}_{\mathrm{mixed}}\leq R$             \\ \hline 
		AltMin \cite{lrpr_best}  & $nr^2\max(r,\log (\frac{1}{\eps}))$ & $mnq\cdot r\log^2(\frac{1}{\eps})$ & RSV  \&  $\sigma_{\n}^2\leq \eps^2 \frac{\norm{\xstar_{k}}^2}{mr^2}$\\ \hline
AltGDmin \cite{lrpr_gdmin}        & $nr^2 \log (\frac{1}{\eps})$  &   ${mnq\cdot r\log(\frac{1}{\eps})}$    & RSV \& $\sigma_v = 0$ (only noise-free case)  \\ \hline
	FedRep \cite{collins2021exploiting}        & $ nr^3\log (n)\log(\frac{1}{\eps})$  &  ${mnq\cdot r\log(\frac{1}{\eps})}$     & RSV \& $\sigma_v = 0$ (only noise-free case)  \\ \hline
	AltMin \cite{netrapalli}                 & $nr^2\log n \max\left(1,\frac{\mathrm{NSR}^2}{r}, \log\left(\frac{1}{\eps}\right),r\frac{\mathrm{NSR}}{\eps^2}\log\left(\frac{1}{\eps}\right)\right)$                              &   $mnq\cdot r\log^2(\frac{1}{\eps})$   & RSV        \\ \hline
		AltMinGD \cite{netrapalli}                 & $nr^2 \log n \max\left(1,\frac{\mathrm{NSR}^2}{r},\log\left(\frac{1}{\eps}\right), r\frac{\mathrm{NSR}}{\eps^2}\log\left(\frac{1}{\eps}\right)\right)$                              &   $mnq\cdot r\log(\frac{1}{\eps})$  & RSV                \\ \hline 
		AltGDmin \cite{lrpr_gdmin_2}      &  $ nr\max\left(r,\log\left(\frac{1}{\eps}\right)\right)$   &  ${mnq\cdot r\log(\frac{1}{\eps})}$ & RSV \& $\sigma_v = 0$ (only noise-free case) \\ 
\hline		
		\textbf{AltGDmin}      &  $ {nr\max\left(r,\log\left(\frac{1}{\eps}\right),\frac{\mathrm{NSR}}{\eps^2}\log\left(\frac{1}{\eps}\right)\right)}$   &  ${mnq\cdot r\log(\frac{1}{\eps})}$ & RSV \\ 
		\textbf{(Proposed)} &&& \\
		\hline
	\end{tabular}
%}
\caption{\small\sl{Comparison with existing work on noisy LRCS. RSV refers to Right Singular Vectors’ (RSV) Incoherence given in Assumption \ref{right_incoh}. We compare sample and time complexity needed to guarantee $\SD_2(\Ustar,\U)\leq \eps$ w.p. at least $1-n^{-10}$. The table treats $\kappa,\mu$ as numerical constants and assumes $n \approx q$.
}}\label{noisy}
\vspace{-0.2in}
\end{table*}

\subsection{Related Work}
Early work on algorithms for LRCS or on its strict generalization, LR phase retrieval (LRPR), includes \cite{hughes_icml_2014,lrpr_tsp}. LRPR involves recovering $\Xstar$ from $\z_k:=|\y_k|$. Thus, any LRPR solution solves LRCS.
The first complete guarantee was provided in \cite{lrpr_icml,lrpr_it} where the authors studied LRPR and provided an alternating minimization (AltMin) based solution.
In parallel work \cite{lee2019neurips}, the LRCS problem was referred to as ``decentralized sketching of LR matrices'' and a convex relaxation based solution was introduced. Solving convex relaxations is known to be very slow. Also, its sample complexity has a bad ($1/\eps^4$) dependence on the final desired error $\eps$. 
Later work on LRPR provided an improved sample complexity guarantee for AltMin and also studied it in the noisy setting \cite{lrpr_best}.
In follow-up work \cite{lrpr_gdmin}, the authors introduced the AltGDmin algorithm, which is a GD based solution that is faster than AltMin. It is also more communication-efficient for a  distributed federated setting. In parallel work (appeared on ArXiv around the same time), the LRCS problem was referred to as ``multi-task linear representation learning'' and algorithms similar to AltGDmin were introduced \cite{collins2021exploiting,netrapalli} (the only difference between these and AltGDmin is the initialization step). The former only studied the noise-free setting, while the latter \cite{netrapalli} studied the noisy case. All these guarantees are compared in Table \ref{noisy}.

\vspace{-0.2in}
\section{AltGDmin algorithm}
Solving the LRCS problem requires solving
\begin{equation}\label{eq:opt}
	\min_{\substack{\U\in\Re^{n\times r}\\ \B \in\Re^{r\times q}}}f(\U,\B): =  \sum_{k=1}^{q}\norm{\y_k-\A_k \U \b_k}^2
\end{equation}
where $\B=[\b_1,...,\b_q]$.
AltGDmin \cite{lrpr_gdmin,collins2021exploiting,lrpr_gdmin_2,netrapalli}
proceeds as follows. We first initialize $\U$ as explained below; this is needed since the optimization problem is clearly non-convex. After this, at each iteration, we alternatively update $\U$ and $\B$ as follows:
(1) Keeping $\U$ fixed, update $\B$ by solving $\min_{\B} f(\U, \B)=\min_{\B} \sum_{k=1}^q \|\y_k - \A_k\U \b_k\|^2$. Clearly, this minimization decouples across columns, making it a cheap least squares problem of recovering $q$ different $r$ length vectors. It is solved as $\b_k \leftarrow (\A_k \U)^\dag \y_k$ for each $k \in [q]$.
(2) Keeping $\B$ fixed, update $\U$ by a GD step, followed by orthonormalizing its columns: 
$\U^+  \leftarrow QR(\U - \eta \nabla_{\U} f(\U,\B)))$. Here $\nabla_{\U} f(\U,\B)=\sum_{k \in [q]}\A_k^{\top} (\A_k \U \b_k - \y_k) \b_k^{\top}$, $\eta$ is the step-size for GD. 
Here we initialize $\U$ by using truncated spectral initialization i.e., computing the top $r$ singular vectors of
\[
\X_0 := \sum_k \A_k^\top (\y_{k})_\trnc \e_k^\top, \  \y_\trnc:=(\y \circ \indic_{|\y| \le  \sqrt\alpha}) 
\]
Here $ \alpha:= 9 \kappa^2 \mu^2 \sum_k \|\y_k\|^2 / mq$.  Here and below, $\y_\trnc$ refers to a truncated version of the vector $\y$ obtained by zeroing out entries of $\y$ with magnitude larger than $\alpha$.

\begin{algorithm}[!htb]
	\caption{\small{The AltGD-Min algorithm.  
}} 
	\label{gdmin}
	\begin{algorithmic}[1]
		%\small
		\STATE {\bfseries Input:} $\y_k, \A_k, k \in [q]$
		\STATE {\bfseries Sample-split:} Partition the measurements and measurement matrices into $2T+1$ equal-sized disjoint sets: one set for initialization and $2T$ sets for the iterations. Denote these by $\y_k^{(\tau)}, \A_k^{(\tau)}, \tau=0,1,\dots 2T$.
		
		\STATE {\bfseries Initialization:}
		\STATE Using $\y_k \equiv \y_k^{(0)}, \A_k \equiv \A_k^{(0)}$,
		set $\alpha = 9 \kappa^2 \mu^2 \frac{1}{mq}\sum_\ik\big|\y_\ik\big|^2$,
		\\
		$\displaystyle  \Xhat_{0}:= (1/m) \sum_{k \in [q]} \A_k^\top \y_{k,trunc}(\alpha) \e_k^\top$
		
		\STATE   Set $\U_0 \leftarrow $ top-$r$-singular-vectors of $\Xhat_0$
		\STATE {\bfseries GDmin iterations:}
		
		\FOR{$t=1$ {\bfseries to} $T$}
		
		\STATE  Let $\U \leftarrow \U_{t-1}$.
		\STATE {\bfseries Update $\b_k, \x_k$: } For each $k \in [q]$, set $(\b_k)_{t}  \leftarrow  (\A_k^{(t)} \U)^\dagger \y_k^{(t)}$ and set $(\x_k)_{t}    \leftarrow \U (\b_k)_{t}$

		\STATE {\bfseries Gradient w.r.t. $\U$: } With $\y_k \equiv \y_k^{(T+t)}, \A_k \equiv  \A_k^{(T+t)}$, compute $\nabla_\U f(\U, \B_t) =   \sum_k \A_k^\top (\A_k \U (\b_k)_t - \y_k) (\b_k)_t^\top$
		\STATE  {\bfseries GD step: } Set $\displaystyle \Uhat^+   \leftarrow \U - (\eta/m) \nabla_\U f(\U, \B_t)$.
		
		\STATE {\bfseries Projection step: }  Compute $\Uhat^+ \qreq \U^+ \bm{R}^+$.
		\STATE Set $\U_t \leftarrow \U^+$.
		\ENDFOR
	\end{algorithmic}
\end{algorithm}

\begin{cor}
	Consider Algorithm \ref{gdmin} and an $0 < \eps < 1$. 
	Set $\eta = c / \sigmax^2$ with a $c \le 0.5$, and $T = C \kappa^2 \log(1/\epsilon)$. 
	If
	\[
 mq \ge C   \kappa^6\mu^2 nr \max\left( \kappa^2 r,\log\left(\frac{1}{\eps}\right), \frac{\mathrm{NSR}}{\eps^2}\log\left(\frac{1}{\eps}\right),\mathrm{NSR}   \right)
\]
	 and $ m \ge C \max (\log q,r)\log (1/\eps)  \max\left(1, \frac{\mathrm{NSR}}{\eps^2 r}\right) ,$, then, w.p. at least $1 - (T+1) n^{-10},$
	\begin{align*}
		& \SD_2(\U_T,\Ustar) \le\eps, \text{ and }  \|\Xhat_T - \Xstar\|_F \le 1.4  \epsilon \|\Xstar\|
	\end{align*}
	\label{gdmin_thm}
\end{cor}
\vspace{-0.2in}
The above result is a direct consequence of the following two results. We provide the proof in Appendix \ref{main_thm}. %of ArXiv version of this work \cite{noisy_lrcs}.

\begin{theorem}[Initialization]
	\label{init_thm}
		Pick a $\delta_0 < 0.1$. If $mq> C\mu^2\kappa^2\left(nr\frac{\kappa^2}{\delta_0^2}+\frac{n}{\delta_0^2}\mathrm{NSR}\right)$, then 
w.p. at least $1 - \exp(-c (n+q))$,
	\[
	\SD_2(\Ustar, \U_0) \le \delta_0.
	\]
\end{theorem}

\begin{theorem}[GD]
	\label{iters_thm}
	If $\SD_2(\Ustar, \U_0) \le \delta_0 = c/\sqrt{r}\kappa^2$, $\eta \le 0.5 /  \sigmax^2$, and if at each iteration $t$,
	$m q\geq\left( C_1 \kappa^4\mu^2(n+r)r + C_2\frac{(n+r)r}{\eps^2}\mathrm{NSR}\mu^2\kappa^6\right)$, and $m \gtrsim \max (\log q,r) \max(1,\frac{\mathrm{NSR}}{\eps^2r })$, 
then w.p. at least $1 - t n^{-10}$, 
	$$
	\SD_2(\Ustar, \U_{t+1}) \le \delta_{t+1}:= \left(1 - 0.5 \eta \sigmin{}^2 \right)^{t}\delta_0 + \eps. 
	$$
\end{theorem}

\subsubsection{Discussion} \textcolor{blue}{Since $\eps<0.1$, $\mathrm{NSR} < \frac{\mathrm{NSR}}{\eps^2}\log\left(\frac{1}{\eps}\right)$. Thus the $\mathrm{NSR}$ factor can be removed from the sample complexity. Suppose that $\mathrm{NSR} \le \eps^2$ so that $\mathrm{NSR}/\eps^2 \le 1$.  In this case, the $\mathrm{NSR}$ value does not matter and the sample complexity is equal to that of the noise-free case \cite{lrpr_gdmin_2} i.e., $\max(\kappa^2 r, \log\left(\frac{1}{\eps}\right))$.} If $\mathrm{NSR} > \eps^2$, then we need a sample complexity that grows as $1/\eps^2$. This can become very large for small $\eps$ values and could even require $m \ge n$ , e.g., for $\eps = 1/\sqrt{n}$ or smaller. A similar dependency is also seen in the result of \cite{netrapalli}. Other work such as that of \cite{lrpr_best} avoids this issue by assuming a bound on the noise. We provide a comparison of our guarantee with all existing ones in Table \ref{noisy}. As can be seen our result improves upon the best existing for the noisy case \cite{netrapalli} by a factor of $ \frac{\max(r,\log(1/\eps)) \log n}{r}$. We extend the result of \cite{lrpr_gdmin_2} to the noisy case. Similar to \cite{lrpr_gdmin_2}, our proof also relies on direct (and easier) techniques \cite{versh_book}: use of sub-Gaussian Hoeffing or sub-exponential Bernstein inequality followed by an easy epsilon-net argument. This is what ensures we can get a tighter bound on the gradient terms. On the other hand, the result of \cite{netrapalli} applies the Hanson-Wright inequality.

\color{blue}

\subsubsection{Experiments}
We simulated $\Ustar$ by orthogonalizing an $n\times r$ standard Gaussian matrix; and the columns $\bstar_k$ were generated i.i.d. from $\mathcal{N}(0,I_r)$. We then set $\Xstar=\Ustar\Bstar$. This was done once (outside Monte Carlo loop). For \( 100 \) Monte Carlo runs, we generated matrices \( \A_k \) for \( k \in [q] \), each entry being i.i.d. standard Gaussian. Similarly, we generated i.i.d. \( \n_k \) for \( k \in [q] \) from $\mathcal{N}(0,\sigma_v^2I_m)$. We then set \( \y_k = \A_k \xstar_k + \n_k \) for \( k \in [q] \). We used step size $\eta=\frac{0.5}{\sigmax^2}$. In Figure \ref{fig_2} we compared our proposed Algorithm \ref{gdmin}, the Method of Moments (MoM) initialization \cite{netrapalli}, and a random initialization. As shown in Figure \ref{fig_2}, our proposed Algorithm \ref{gdmin} converges faster than both MoM and random initialization.
\begin{figure}[htbp]
	\vspace{-0.15in}
	\begin{center}
		\includegraphics[width=0.5\textwidth, height=0.6\linewidth]{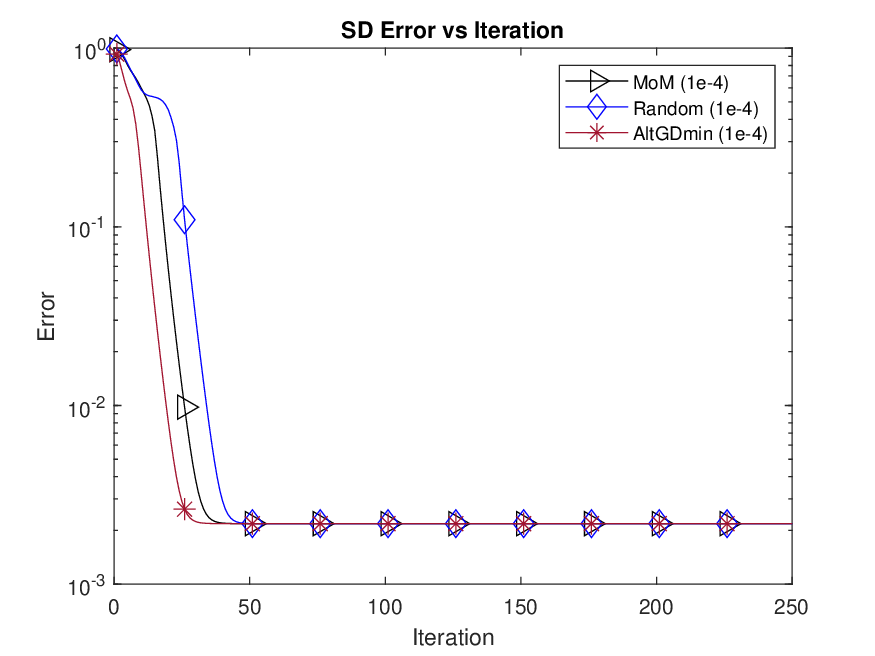}
	\end{center}
	\vspace{-0.5cm}
	\caption{\small $\SD_2(\U_t,\Ustar)$ vs Iteration $t$ with $n=600$, $m=30$, $q=600$, $r=4$, and $\sigma_v^2=10^{-4}$. We compared our proposed Algorithm \ref{gdmin} (AltGDmin), the Method of Moments (MoM) initialization \cite{netrapalli}, and a random initialization.}\label{fig_2}
\end{figure}

\color{black}

\vspace{-0.3in}
\section{Proof Outline}
The proof consists of two parts: first, we show that the distance between $\Ustar$ and $\U$ decreases each iteration, up to a term depending on $\mathrm{NSR}=\frac{q\sigma_{\n}^2}{\sigmin{}^2}$. Second, truncated spectral initialization provides a good initial $\U_0$ with $\SE(\U_0,\Ustar) = \delta_0 = 0.1/\sqrt{r}\kappa^2$.

\vspace{-0.2in}
\subsection{Proof outline for GD step Theorem \ref{iters_thm}}
We first analyze the minimization step for updating $\B$ and then the GD step w.r.t. $\U$.

Let  $\U^+:= \U_{t+1}$,  $\U:= \U_{t}, \B:=\B_{t+1}$, and $\X:= \U \B$.

\subsubsection{Analyzing the minimization step for updating \textbf{B}}
Let $\g_k:= \U^\top \xstar_k = (\U^\top \Ustar) \bstar_k$. We can write
\begin{align}
	\|\b_{k} - \g_k\| &\leq \|\left(\U{}^\top\A_k^\top\A_k\U\right)^{-1}\| \nonumber\\& \times ~\left( \|\U{}^\top\A_k^\top \A_k(\I-\U\U{}^\top)\xstar_{k}\| +  \|\U{}^\top\A_k^\top\n_k\| \right).
	\label{gk_bhatk_bnd_1}
\end{align}
Using proof technique from \cite[Lemma 3.3 part 1]{lrpr_gdmin} i.e., sub-exponential Bernstein inequality followed by an epsilon-net argument, one can bound each term in \eqref{gk_bhatk_bnd_1} and show the following

\begin{lemma}[Bound on $\|\g_k - \b_{k}\|$]
\label{Blemma_new2}	
	Assume that $\SD_2(\U,\Ustar) \le \delta_t$. 
	Recall that $\g_k:= \U^\top \xstar_k = (\U^\top \Ustar) \bstar_k$.
	If $m\gtrsim \max\left(\log q,r, \max(\log q,r) \frac{\mathrm{NSR}}{\eps^2 r }\right)$, then, w.p. at least $1-n^{-10}$, 
	\[\|\g_k - \b_{k}\|   \leq 0.4\delta_t \|\bstar_k\| 
\]
\end{lemma}
 We provide the proof in Appendix \ref{gd_thms}. % of ArXiv version of this work \cite{noisy_lrcs}.

Once we have a bound on $\|\b_k - \g_k\|$, using $\SE_2(\U,\Ustar) \le \delta_t$, we can prove \cite[Theorem 4]{lrpr_gdmin_2} as done there. We state it below.
\begin{lemma}[{\cite[Theorem 4]{lrpr_gdmin_2}}]
	Assume that $\SE_2(\U,\Ustar) \le \delta_t$. If $\delta_t \le c / \sqrt{r} \kappa^2$,
	and 	if $m$ satisfies the bound given in Lemma \ref{Blemma_new2}	 
	\ben
	\item $\|\b_k - \g_k\| \le 0.4\delta_t \|\bstar_k\| \le 0.4 \delta_t \mu\sqrt{\frac{r}{q}}\sigmax$,
\item $\|\x_k - \xstar_k\| \le 1.4  \delta_t \mu\sqrt{\frac{r}{q}}\sigmax$ and  $\|\b_k\| \le 1.1\mu\sqrt{\frac{r}{q}}\sigmax$

	\item $\|\B - \G\|_F \le  0.4 \sqrt{r} \delta_t \sigmax $, $\|\X - \Xstar\|_F \le 1.4 \sqrt{r} \delta_t \sigmax $,

	\item $\sigmamin(\B) \ge  0.8 \sigmin$, $\sigma_{\max}(\B) \le  1.1 \sigmax$
	\een
	(only item 4 bounds require the upper bound on $\delta_t$).
	\label{Blemma_new1}
\end{lemma}

\subsubsection{Analyzing the GD step for updating \textbf{U}}
Observe that
\[
\E[\nabla f(\U,\B)]=m(\X-\Xstar)\B^\top=m(\U\B\B^\top-\Ustar\Bstar\B^\top)
\]
Let  $\nabla f:= \nabla f(\U,\B)$ and
\[
	\Err=\nabla f-\E[\nabla f] = \nabla f - m(\U\B\B^\top-\Ustar\Bstar\B^\top)
\]
Recall that $\U^+ := QR(\U - (\eta/m) \nabla f)$ denotes the updated $\U$ after one GD step followed by orthonormalization.
We state the algebra lemma from \cite{lrpr_gdmin_2} next.
\begin{lemma}(algebra lemma)\label{sd_pop}
	 We have
	\begin{align*}
		\SD_2(\Ustar,\U^+)
		\leq \frac{ \| \I_{r}-\eta\B \B^\top \| \SD_2(\Ustar,\U)  + \frac{\eta}{m} \| \Err\| }{ 1-\frac{\eta }{m}\|\E[\nabla f(\U,\B)] \|  -\frac{\eta }{m} \|\Err\|}
	\end{align*}
\end{lemma}
Below we bound $\Err$ term
\begin{small}
	\begin{align*}
		\Err: =	&\nabla f(\U,\B)-\E[\nabla f(\U,\B)] \\ %&=\sum_k \A_k^\top(\A_k\x_k-\y_k)\b_k^\top - \sum_k m (\x_k-{\xstar_k})\b_k{}^\top \\
		%& = \sum_k \A_k^\top(\A_k\x_k-\A_k\xstar_k - \n_k)\b_k^\top - \sum_k m (\x_k-{\xstar_k})\b_k{}^\top \\
		%& = \sum_k \A_k^\top(\A_k\x_k-\A_k\xstar_k)\b_k^\top - \sum_k m (\x_k-{\xstar_k})\b_k{}^\top - \sum_k \A_k^\top\n_k\b_k^\top\\
		& = \sum_k \A_k^\top(\A_k\x_k-\A_k\xstar_k)\b_k^\top - \E[\nabla f]  - \sum_k \A_k^\top\n_k\b_k^\top\\
		& := \mathrm{Term}_{\mathrm{nonoise}} - \mathrm{Term}_{\mathrm{noise}}
	\end{align*}
\end{small}
where $\mathrm{Term}_{\mathrm{nonoise}} = \sum_k \A_k^\top(\A_k\x_k-\A_k\xstar_k)\b_k^\top - \E[\nabla f]$, and $\mathrm{Term}_{\mathrm{noise}}=\sum_k \A_k^\top\n_k\b_k^\top$. 
$\mathrm{Term}_{\mathrm{nonoise}}$ is bounded in \cite[Lemma 3.5 item 1]{lrpr_gdmin}, and we can bound $\mathrm{Term}_{\mathrm{noise}}$ using sub-exponential Bernstein inequality followed by an epsilon-net argument and can conclude the following lemma.

\begin{lemma} \label{err_bound_lemma}
	Assume $\SD_2(\Ustar,\U)\leq \delta_t<\delta_0$. Then,
	w.p. at least $1-  \exp ( C(n+r) -c\frac{\eps_1^2 m q }{\kappa^4 \mu^2 r} )-\exp (C(n+r)- c \frac{\eps_2^2 m q\sigmin^2 }{ \sigma_{\n}^2 q\kappa^2 } )- 2\exp(\log q + r - c m) - \exp (\log q + r-\frac{m\eps^2\mu^2r\sigmax{}^2}{Cq\sigma_{\n}^2} )$ 
	\[
	\norm{\nabla f(\U,\B)-\E[\nabla f(\U,\B)]}\leq \eps_1 \delta_t m \sigmin^2+ \eps_2 m  \sigmin^2
	\]
Also, using Lemma \ref{Blemma_new1},
$\|\E[\nabla f(\U,\B)]\| \le m \|\X- \Xstar\|_F \|\B\| \le Cm\sqrt{r}\delta_t\sigmax{}^2 $.
\end{lemma}
 We provide the proof in Appendix \ref{gd_thms}. %of ArXiv version of this work \cite{noisy_lrcs}.

\subsubsection{Proof of Theorem \ref{iters_thm}}\label{gd_proof_sec} 
Using the bounds from Lemma \ref{err_bound_lemma}, \ref{Blemma_new1} and setting $\eps_1 =0.1$, and $\eps_2=\frac{\eps c}{4\sqrt{r}\kappa^2}$, ($\eps$ is the final desired error) we conclude the following: if in each iteration, $m q\geq\left( C_1 \kappa^4\mu^2(n+r)r + C_2\frac{(n+r)r}{\eps^2}\mathrm{NSR}\kappa^6\right)$, $m \gtrsim \max\left(\log q,r, \max(\log q , r) \frac{\mathrm{NSR}}{\eps^2\mu^2r\kappa^2}\right)$, then, w.p. $1-n^{-10}$ 
\begin{align}
	&\left(1-\frac{\eta }{m}\|\E[\nabla f(\U,\B)] \|  -\frac{\eta }{m} \|\Err\|\right)^{-1} \\ &\leq \left(1 - \eta C\sqrt{r}\delta_t\sigmax{}^2-\eta\sigmin{}^2(\eps_1\delta_t+\eps_2)\right)^{-1} \notag \\
	& \leq (1 + 0.6 \eta \sigmin{}^2) \label{eq:denom}
\end{align}
We used $\delta_t<\delta_0=\frac{c}{C\sqrt{r}\kappa^2}$, $c<0.1$, $\eps_2=\frac{\eps c}{4\sqrt{r}\kappa^2}<0.1$ and  $(1-x)^{-1}< (1+2x)$ if $|x|<1$.
 	
Using Lemma \ref{Blemma_new1}, we get
\[
\lambda_{\min}(\I_{r} -  \eta  \B \B^\top) = 1 - \eta  \|\B\|^2 \ge 1 - 1.2 \eta  \sigmax^2
\]
Thus, if $\eta \le 0.5/  \sigmax^2 $, then the above matrix is p.s.d. This along with Lemma \ref{Blemma_new1} implies that
\[
\|\I_{r} -  \eta  \B \B^\top\| = \lambda_{\max}(\I -  \eta  \B \B^\top) \le 1 - 0.8 \eta   \sigmin^2
\]
Using the above and Lemmas \ref{err_bound_lemma} and \ref{Blemma_new1},
\begin{align}
	&\| \I_{r}-\eta\B \B^\top \| \SD_2(\Ustar,\U)  + \frac{\eta}{m} \| \Err\| \\&\leq (1 - 0.8 \eta \sigmin{}^2 ) \delta_t + 0.1 \eta \sigmin{}^2 \delta_t+ \eta \eps_2\sigmin{}^2 \notag \\
	& =  (1 - 0.7 \eta \sigmin{}^2 ) \delta_t +  \eta \sigmin{}^2  \eps_2 \label{eq:num}
\end{align} 	
Combining \eqref{eq:num}, \eqref{eq:denom} we get
\begin{small}
	\begin{align}
		&\SD_2(\Ustar,\U^+)\\
		&\leq (1 - 0.7 \eta \sigmin{}^2 )(1 + 0.6 \eta \sigmin{}^2) \delta_t + \eta \sigmin{}^2\eps_2(1 + 0.6 \eta \sigmin{}^2) \notag \\
		& \leq (1 - 0.1 \eta \sigmin{}^2 ) \delta_t + 2 \eta \sigmin{}^2\eps_2
	\end{align}	
\end{small}
Denote $  a=(1 - 0.1 \eta \sigmin{}^2 )$, $B = 2 \eta \sigmin{}^2\eps_2$.
Now using induction if $\delta_{1} \leq a\delta_0 + B$, and $\delta_{t}\leq a\delta_{t-1} + B$ implies $\delta_t \leq a^t \delta_0 + B\sum_{\tau=0}^{t-1}a^\tau$	
Since $a<1$ implies using sum of geometric series we get $\delta_{T}\leq a^T \delta_0 + B \frac{1}{1-a}$. %\[\delta_{T}\leq a^T \delta_0 + B \frac{1-a^T}{1-a}\]
That is
\begin{small}
	\begin{align*}
		&\SD_2(\Ustar, \U_{T})  \le  \left(1 - 0.1 \eta \sigmin{}^2 \right)^{T}\delta_0 + 2 \eta \sigmin{}^2\eps_2\frac{1}{1-a} \\
		%\SD_2(\Ustar, \U_{T})  &\le  \left(1 - 0.5 \eta \sigmin{}^2 \right)^{T}\delta_0 + 2\eps_2 \\
		&\le  \left(1 - 0.1 \eta \sigmin{}^2 \right)^{T}\delta_0 + \frac{1}{0.1} 2  \eps_2 
		\le  \left(1 - 0.1 \eta \sigmin{}^2 \right)^{T}\delta_0 + C\frac{\eps }{2\sqrt{r}\kappa^2}\\
		& \le  \left(1 - 0.1 \eta \sigmin{}^2 \right)^{T}\delta_0 + \eps
	\end{align*}
\end{small}	
w.p. at least $1-  \exp ( C(n+r) -c\frac{ m q }{\kappa^4 \mu^2 r} )-\exp (C(n+r)- c \frac{ \eps^2m q\sigmin^2 }{ r\sigma_{\n}^2 q\kappa^6 } )- 2\exp(\log q + r - c m) - \exp (\log q + r-\frac{m\eps^2\mu^2r\sigmax{}^2}{Cq\sigma_{\n}^2} )$.  	

\subsection{Proof outline for Initialization Theorem \ref{init_thm}}
 We need a new result to show $\E[\Xhat_{0}|\alpha]= \Xstar \D(\alpha) $ where $\D$ is a diagonal matrix defined in Lemma \ref{X_0} given below.  
\begin{lemma} \label{X_0}
	Conditioned on $\alpha$, we have the following conclusions.
	 Let $\bzeta$  be a scalar standard Gaussian r.v.. Define
	\begin{small}
	\[
	\w_k(\alpha)=\E\left[\bzeta^2 \indic_{ \left\{ |\bzeta| \leq \frac{\sqrt{\alpha}}{\sqrt{\|\xstar_{k}\|^2+\sigma_{\n}^2}} \right\} } \right].
	\]
\end{small}
	Then,
	\begin{align}
		&\E[\Xhat_0|\alpha] = \Xstar \D(\alpha), \nonumber\\
		&\text{ where }  \D(\alpha):=diagonal(\w_k(\alpha),k \in [q])
		\label{X0}
	\end{align}
	i.e. $\D(\alpha)$ is a diagonal matrix of size $q\times q$ with diagonal entries $\w_k(\alpha)$ defined above.
\end{lemma}

\begin{proof}
It suffices to show that $\E[(\Xhat_0)_k |\alpha] = \xstar_k \w_k(\alpha)$ for each $k$. Using the same argument given in proof of \cite[Lemma 3.6 item 1]{lrpr_gdmin} we can write
\begin{align*}
	(\Xhat_0)_k = 
\frac{1}{m} \sum_i \Q_k  \tilde\a_\ik \s \indic_{ | \s| \le \sqrt{\alpha} }
\end{align*} where $\Q_k$ is an $n \times n$ unitary matrix with first column $\xstar_k/\|\xstar_k\|$, $\tilde\a_\ik:=\Q_k^\top \a_\ik$ has the same distribution as $\a_\ik$, both are $\mathcal{N}(0,\I_n)$, $\s=\|\xstar_k\|\tilde\a_\ik (1) + \n_{\ik}$.
Thus \[
\E[(\Xhat_0)_k] = \Q_k\e_1 \E[\tilde\a_\ik(1)\s\indic_{|\s| < \sqrt{\alpha}}].
\]

This follows because $\{\tilde\a_\ik(i)\}_{i=2}^{m}$ are independent of $\s$. 
Removing the subscripts, the problem reduces to finding the expectation $\E[\a(1)\s\indic_{ |\s| < \sqrt{\alpha} }]$ where $\s=\a(1)\|\xstar_k\| +  \n$, $\a(1), \n$ are independent normal random variables, $\|\xstar_k\|, \alpha$ are constants.
Using law of total expectation
\begin{align*}
	\E[\a(1)\s\indic_{ |\s| < \sqrt{\alpha} }] &= \E[\E[\a(1)\s\indic_{ |\s| < \sqrt{\alpha} }|\s = s]]\\
	%&\text{Law of total expectation}\\
	&=\E[\s\indic_{ |\s| < \sqrt{\alpha} }\E[\a(1)|\s = s]]
\end{align*}

We solve the inner expectation first as follows. Since $\a(1), \s$ are joint Gaussian, thus, using the formula for conditional expectation of joint Gaussians\footnote{If $X,Y$ joint Gaussian, then
$\E[X|Y] = \E[X]+ cov(X,Y) \frac{1}{var(Y)} (Y - \E[Y])$. For our setting $X= \zeta$, $Y = \s = \zeta\|\xstar_k\| + \bm{v}$,  $\E[X]=0$, $\E[Y]=0$, $cov(X,Y) = E[\zeta^2]\|\xstar_k\| = \|\xstar_k\|$ and $var(Y) = \|\xstar_k\|^2 + \sigma_v^2$
}, we can obtain an exact expression for $\E[\a(1)|\s = s]$:
\[
\E[\a(1)|\s = s] =   \frac{s\|\xstar_k\|}{\|\xstar_k\|^2+\sigma_{\n}^2} 
\]
Now taking the outer expectation
\begin{small}
	\begin{align*}
		&\E\left[\frac{\s^2\|\xstar_k\|}{\|\xstar_k\|^2+\sigma_{\n}^2}\indic_{ |\s| < \sqrt{\alpha} }\right] \\ &=\|\xstar_k\|\E\left[\left(\frac{\s}{\sqrt{\|\xstar_k\|^2+\sigma_{\n}^2}}\right)^2\indic_{ \left(\frac{\s}{\sqrt{\|\xstar_k\|^2+\sigma_{\n}^2}}\right)^2 < \frac{\alpha}{\|\xstar_k\|^2+\sigma_{\n}^2}}\right]
	\end{align*}
\end{small}

Note $\s=\a(1)\|\xstar_k\| + \n \sim\mathcal{N}(0,\|\xstar_k\|^2+\sigma_{\n}^2)$ implies $\bzeta=\frac{\s}{\sqrt{\|\xstar_k\|^2+\sigma_{\n}^2}}\sim\mathcal{N}(0,1)$. Implies $\E[(\Xhat_0)_k] = \Q_k\e_1 \|\xstar_k\|\w_k(\alpha)=\xstar_k\w_k(\alpha)$. 
\end{proof}

\subsubsection{Proof of Theorem \ref{init_thm}}\label{init_proof_sec}
The proof technique is similar to that of \cite[Lemma 3.6]{lrpr_gdmin} where we apply Wedin's $\sin \Theta$ theorem on $\Xhat_0$ and $\E[\Xhat_0|\alpha]$ to bound $\SD_2(\U_0, \Ustar)$. 
 Then using sub-Gaussian Hoeffding inequality followed by an epsilon-net argument, we bound the terms in $\SD_2(\U_0,\Ustar)$. The bounds on these terms are a sum of two terms, the first is the same as that in \cite[Lemma 3.8]{lrpr_gdmin} and the second is $1.1\eps_1\sqrt{q}\sigma_{\n}$.

We can conclude, for a $\delta_0 < 0.1$, $\SD_2(\U_0,\Ustar) < \delta_0$ w.p. at least $1 - 5 \exp(-c(n+q))$ if $mq > C \kappa^2 \mu^2 (n+q)(r\sigmax{}^2 + q\sigma_{\n}^2) / \delta_0^2\sigmin{}^2$. The proof  is provided in Appendix \ref{init_thms}.
%We can conclude, for a $\delta_0 < 0.1$, $\SD_2(\U_0,\Ustar) < \delta_0$ w.p. at least
%\\ $  1 - 2\exp( (n+q)- c mq \delta_0^2\sigmin{}^2 /\mu^2\kappa^2(r\sigmax{}^2 + q\sigma_{\n}^2)  ) - 2\exp( n r - c mq \delta_0^2\sigmin{}^2 / \mu^2\kappa^2(r\sigmax{}^2 + q\sigma_{\n}^2) ) - 2\exp( q r - c mq \delta_0^2\sigmin{}^2 / \mu^2\kappa^2(r\sigmax{}^2 + q\sigma_{\n}^2) ) - \exp(- c mq \delta_0^2\sigmin{}^2  /\mu^2\kappa^2(r\sigmax{}^2 + q\sigma_{\n}^2) ) $.
%
%This is $\ge 1 - 5 \exp(-c(n+q))$ if $mq > C \kappa^2 \mu^2 (n+q)(r\sigmax{}^2 + q\sigma_{\n}^2) / \delta_0^2\sigmin{}^2$.
%
%The proof  is provided in Appendix \ref{init_thms}. % of the arXiv version of this document \cite{noisy_lrcs}.

		\clearpage
\bibliographystyle{IEEEtran}
	\bibliography{byz}

	\clearpage
	\appendices \renewcommand\thetheorem{\Alph{section}.\arabic{theorem}}

\section{Proof of Theorem \ref{gdmin_thm}}\label{main_thm}

\begin{proof}
	The $\SD_2(.)$ bound is an immediate consequence of Theorems \ref{init_thm} and \ref{iters_thm}. To apply Theorem \ref{iters_thm}, we need $\delta_0 = c / \sqrt{r}\kappa^2$. By Theorem \ref{init_thm}, if $m_0 q > C \left(\mu^2\kappa^8nr^2 + \mu^2\kappa^6nr\mathrm{NSR}\right)$, then, w.p. at least $1-n^{-10}$, $\SD_2(\Ustar, \U_0) \le \delta_0 = c / \sqrt{r}\kappa^2$.
	With this,  if, at each iteration, $m_t q\geq C\left(\kappa^4\mu^2nr +\mu^2\kappa^6 \frac{nr}{\eps^2}\mathrm{NSR}\right) \text{ and } m_t \gtrsim  \max\left(\log q,r, \frac{(\log q + r)}{\eps^2\mu^2r\kappa^2}\mathrm{NSR}\right)$, then by Theorem \ref{iters_thm}, w.p. at least $1-(t+1) n^{-10}$,
	\[\SD_2(\Ustar, \U_{t+1}) \leq \left(1 - \frac{c_{\eta}}{\kappa^2}\right)^{t}\frac{c}{\sqrt{r}\kappa^2} + \eps/2\] where $c_{\eta}=0.5c$.
	To guarantee $\SD_2(\U_T, \Ustar) \le \eps$, we need $T \ge C \frac{\kappa^2}{c_\eta} \log(1 / \eps ).$
	This follows by using $\log(1-x) < - x$ for $|x|<1$ and using $\kappa^2 \sqrt{r} \ge 1$.
	Thus, setting $c_\eta = 0.4$, our sample complexity $m = m_0 + m_1 T$ becomes
	$
	mq \ge C \mu^2\kappa^6nr\left(\kappa^2r + \mathrm{NSR} + \log\left(\frac{1}{\eps}\right)+\kappa^2\frac{\mathrm{NSR}}{\eps^2}\log\left(\frac{1}{\eps}\right)\right),
	$
	and $ m \ge C \max\left(\log q,r, \frac{(\log q + r)}{\eps^2\mu^2r\kappa^2}\mathrm{NSR}\right) \log (1/\eps)$.
	
	Thus by setting $T =  C \kappa^2 \log(1/\epsilon)$ in this, we can guarantee $\SD_2(\U_T, \Ustar)\leq \eps/2 + \eps/2 =\eps $. 
	This proves the $\SD_2(\U_T, \Ustar)$ bound. 
\end{proof}

\section{Proof of the GD result}\label{gd_thms}

\begin{lemma}[Bound on $\|\g_k - \b_{k}\|$]
	Assume that $\SD_2(\U,\Ustar) \le \delta_t\leq \delta_0 = \frac{c}{\sqrt{r}\kappa^2} $.
	Let $\g_k:= \U^\top \xstar_k = (\U^\top \Ustar) \bstar_k$.
	
	If $m\gtrsim \max\left(\log q,r, \frac{(\log q + r)}{\eps^2\mu^2r\kappa^2}\mathrm{NSR}\right)$, 
	then, w.p. at least $1-n^{-10}$, %whp\textsuperscript{\ref{note1}}.
	\[\|\g_k - \b_{k}\|   \leq 0.4\delta_t\mu\sqrt{\frac{r}{q}}\sigmax\]
\end{lemma}
\begin{proof}
	We bound $\|\g_k - \b_{k}\|$ here.
	Recall that $\g_k = \U^\top \xstar_k$.
	Since $\y_k = \A_k\xstar_{k} + \n_k = \A_k\U\U{}^\top\xstar_{k} + \A_k(\I-\U\U{}^\top)\xstar_{k} + \n_k$, therefore
	\begin{align*}
		\b_{k}  &= \left(\U{}^\top\A_k^\top\A_k\U\right)^{-1}(\U{}^\top\A_k^\top)\A_k\U\U{}^\top\xstar_{k} \\&\qquad+ \left(\U{}^\top\A_k^\top\A_k\U\right)^{-1}(\U{}^\top\A_k^\top) \A_k(\I-\U\U{}^\top)\xstar_{k}\\
		& \qquad +\left(\U{}^\top\A_k^\top\A_k\U\right)^{-1}(\U{}^\top\A_k^\top)\n_k\\
		&=\left(\U{}^\top\A_k^\top\A_k\U\right)^{-1}\left(\U{}^\top\A_k^\top\A_k\U\right)\U{}^\top\xstar_{k}  \\&\qquad+ \left(\U{}^\top\A_k^\top\A_k\U\right)^{-1}(\U{}^\top\A_k^\top) \A_k(\I-\U\U{}^\top)\xstar_{k} \\
		& \qquad + \left(\U{}^\top\A_k^\top\A_k\U\right)^{-1}(\U{}^\top\A_k^\top)\n_k\\
		&=\g_k + \\& \left(\U{}^\top\A_k^\top\A_k\U\right)^{-1}\left((\U{}^\top\A_k^\top) \A_k(\I-\U\U{}^\top)\xstar_{k} + (\U{}^\top\A_k^\top)\n_k\right).
	\end{align*}
	Thus,
	\begin{align}
		\|\b_{k} - \g_k\| &\leq \|\left(\U{}^\top\A_k^\top\A_k\U\right)^{-1}\| \nonumber\\& \qquad\times ~\left( \|\U{}^\top\A_k^\top \A_k(\I-\U\U{}^\top)\xstar_{k}\| +  \|\U{}^\top\A_k^\top\n_k\| \right).
		\label{gk_bhatk_bnd}
	\end{align}
	Using standard results from \cite{versh_book}, one can show the following:
	\ben
	\item
	W.p. $\ge 1-q\exp\left(r-cm\right)$, for all $k\in[q]$, $\min_{\w \in \S_r} \sum_i \big|\a_\ik{}^\top\U\w\big|^2 \ge 0.7 m$ and so
	\begin{align*}
		\|\left(\U{}^\top\A_k^\top\A_k\U\right)^{-1}\| &= \frac{1}{\sigma_{\min}\left(\U{}^\top\A_k^\top\A_k\U \right)} \\&= \frac{1}{\min_{\w\in\S_{r}} \sum_i \langle \U^\top \a_\ik , \w \rangle^2 }\\& \le \frac{1}{0.7 m}
	\end{align*}		
	\item W.p. at least $1-q\exp(r-cm)$, $ \forall k\in[q]$,
	\[
	\|\U{}^\top\A_k^\top \A_k(\I-\U\U{}^\top)\xstar_{k}\| \leq    0.15 m \| (\I-\U\U{}^\top)\xstar_{k}\|
	\]
	\item W.p. at least $1- q \exp(r -\frac{m\eps_2^2}{C\sigma_{\n}^2}), \forall k\in[q]$,
	\[
	\|\U{}^\top\A_k^\top\n_k\| \leq    m\eps_2
	\]
	\een
	Combining the above three bounds and \eqref{gk_bhatk_bnd}, w.p. at least $1 - 2\exp(\log q + r - c m) - \exp (\log q + r-\frac{m\eps_2^2}{C\sigma_{\n}^2} )$, $\forall k\in[q]$,
	\begin{align*}
		\|\g_k - \b_{k}\|   &\leq c_1 \|\left(\I_n-\U\U^\top\right)\Ustar\bstar_k\| + c_2\eps_2\\
		& \leq c_1 \delta_t\|\bstar_k\| + c_2\eps_2\\
		&\leq c_1 \delta_t\mu\sqrt{\frac{r}{q}}\sigmax + c_2\eps
	\end{align*}
	
	We have used Right Singular Vectors’ Incoherence Assumption \ref{right_incoh} to bound $\|\bstar_k\|$. Setting $\eps_2=\delta_t\mu\sqrt{\frac{r}{q}}\sigmax  $ we get w.p. at least $1 - 2\exp(\log q + r - c m) - \exp (\log q + r-\frac{m\delta_t^2\mu^2r\sigmax{}^2}{Cq\sigma_{\n}^2} )$, $\forall k\in[q]$
	
	\begin{align*}
		\|\g_k - \b_{k}\|   &\leq c_3\delta_t\mu\sqrt{\frac{r}{q}}\sigmax
	\end{align*}
	
	Since we are showing exponential decay, for a $t < T$, $\delta_t \ge  \delta_T=\eps$.  Therefore using $\delta_t\ge\eps$ where $\eps$ is the final desired error, in the probability we get if $m\geq \log q +r+ \frac{(\log q + r)}{\eps^2\mu^2r\kappa^2}\mathrm{NSR}$, then, whp,
	
	\[\|\g_k - \b_{k}\|   \leq 0.4\delta_t\mu\sqrt{\frac{r}{q}}\sigmax\]
	
	\textbf{Bounding $ \|\U{}^\top\A_k^\top\n_k\|$}
	
	Notice that
	\begin{align*}
		&\|\U{}^\top\A_k^\top\n_k\| \\
		&\qquad= \max_{\w\in\S_{r}} \w{}^\top\U{}^\top\A_k^\top \n_k \\&\qquad=  \max_{\w\in\S_{r}} \sum_i \a_\ik{}^\top\n_\ik\U\w
	\end{align*}
	Clearly $\E\left[ \U{}^\top\A_k^\top\n_k\right] =  0$.
	
	Moreover, the summands are products of sub-Gaussian r.v.s and are thus sub-exponential with sub-exponential norm $ K_{i} = C \sigma_{\n} $. Also, the different summands are mutually independent and zero mean. 
	Applying  sub-exponential Bernstein for a fixed $\w\in\S_{r}$ and with $t=m\eps_2$,
	
	\[\frac{t^2}{\sum_i K_{i}^2}\leq \frac{t^2}{mC\sigma_{\n}^2}=\frac{m\eps_2^2}{C\sigma_{\n}^2}\]
	\[\frac{t}{\max_i K_{i}}\leq \frac{t}{C\sigma_{\n}}=\frac{m\eps_2}{C\sigma_{\n}}\]
	For $\eps_2<\sigma_{\n}$ the first term above is smaller therefore
	$$ | \sum_i \a_\ik{}^\top\n_\ik\U\w | \leq m\eps_2$$
	
	w.p. at least $1-\exp (-\frac{m\eps_2^2}{C\sigma_{\n}^2} )$.
	Using epsilon net argument the above term is bounded by $m\eps_2$ for all $\w \in \S_r$ w.p. at least $1-\exp(r -\frac{m\eps_2^2}{C\sigma_{\n}^2})$. Using a union bound over all $q$ columns, the bound holds for all $q$ columns w.p.  at least $1- q \exp(r -\frac{m\eps_2^2}{C\sigma_{\n}^2})$.
	
\end{proof}

\begin{lemma}
	Assume $\SD_2(\Ustar,\U)\leq \delta_t<\delta_0$. Then,
	w.p. at least $1-  \exp ( C(n+r) -c\frac{\eps_1^2 m q }{\kappa^4 \mu^2 r} )-\exp (C(n+r)- c \frac{\eps_2^2 m q\sigmin^2 }{ \sigma_{\n}^2 q\kappa^2 } )- 2\exp(\log q + r - c m) - \exp (\log q + r-\frac{m\eps^2\mu^2r\sigmax{}^2}{Cq\sigma_{\n}^2} )$ 
	\[
	\norm{\nabla f(\U,\B)-\E[\nabla f(\U,\B)]}\leq \eps_1 \delta_t m \sigmin^2+ \eps_2 m  \sigmin^2
	\]
\end{lemma}
\begin{proof}
	\begin{align*}
		&\nabla f(\U,\B)-\E[\nabla f(\U,\B)]\\&=\sum_k \A_k^\top(\A_k\x_k-\y_k)\b_k^\top - \sum_k m (\x_k-{\xstar_k})\b_k{}^\top \\
		& = \sum_k \A_k^\top(\A_k\x_k-\A_k\xstar_k - \n_k)\b_k^\top - \sum_k m (\x_k-{\xstar_k})\b_k{}^\top \\
		& = \sum_k \A_k^\top(\A_k\x_k-\A_k\xstar_k)\b_k^\top - \sum_k m (\x_k-{\xstar_k})\b_k{}^\top \\&\qquad\qquad- \sum_k \A_k^\top\n_k\b_k^\top\\
		& = \sum_k \A_k^\top(\A_k\x_k-\A_k\xstar_k)\b_k^\top - \E[\nabla f]  - \sum_k \A_k^\top\n_k\b_k^\top\\
		& = \mathrm{Term_{\mathrm{nonoise}}} - \mathrm{Term_{\mathrm{noise}}}
	\end{align*}
	where $\mathrm{Term_{\mathrm{nonoise}}} = \sum_k \A_k^\top(\A_k\x_k-\A_k\xstar_k)\b_k^\top - \E[\nabla f]$, and $\mathrm{Term_{\mathrm{noise}}}=\sum_k \A_k^\top\n_k\b_k^\top$
	
	Proof steps for bounding $\norm{\mathrm{Term_{\mathrm{nonoise}}}}$ are similar to those given \cite[Lemma 5 item 3]{lrpr_gdmin_2} with only change in probability expression i.e.,  w.p. at least $1-  \exp ( C(n+r) -c\frac{\eps_1^2 m q }{\kappa^4 \mu^2 r} )- 2\exp(\log q + r - c m) - \exp (\log q + r-\frac{m\eps^2\mu^2r\sigmax{}^2}{Cq\sigma_{\n}^2} )$,
	\[
	\|\mathrm{Term_{\mathrm{nonoise}}}\| \le  \eps_1 \delta_t m \sigmin^2
	\]
	
	\textbf{Bounding $\norm{\mathrm{Term_{\mathrm{noise}}}}$}
	
	For a fixed $\z\in\S_n, \w\in\S_r$ we have
	\begin{align*}
		&\z{}^\top \sum_k \A_k^\top\n_k\b_k^\top \w \\
		& \qquad = \sum_\ik\left(\a_\ik{}^\top\n_{\ik}\z\right)\left(\w{}^\top\b_{k}\right) \\
	\end{align*}
	
	Observe that the summands are independent, zero mean, sub-exponential r.v.s with sub-exponential norm $K_\ik \le C \| \w \| \cdot |\z^\top \b_k| \cdot \sigma_{\n} =C |\z^\top \b_k| \cdot \sigma_{\n} $.
	We apply the sub-exponential Bernstein inequality, Theorem 2.8.1 of \cite{versh_book}, with $t = \eps_2 m  \sigmin^2$. We have
	\begin{align*}
		\frac{t^2}{\sum_\ik K_\ik^2}  %\max_{g' \neq g} \max_{k \in \S_{g'}}
		& \ge \frac{\eps_2^2 m^2   \sigmin^4}{ m \sigma_{\n}^2 \sum_k (\z^\top \b_k)^2   }\\
		& = \frac{\eps_2^2 m\sigmin^4}{\sigma_{\n}^2 \|\z^\top \B\|^2 }\\
		& \ge \frac{\eps_2^2 m\sigmin^4 }{C\sigma_{\n}^2 \sigmax^2 }\\
		& = \frac{\eps_2^2 m \sigmin^2 }{ C\sigma_{\n}^2 \kappa^2 }
	\end{align*}

	\begin{align*}
		\frac{t}{\max_\ik K_\ik} & \ge  \frac{\eps_2  m  \sigmin^2 }{C\sigma_{\n}\max_k \| \b_k\|  }
		\ge  \frac{\eps_2 m \sqrt{q} \sigmin}{C\kappa\sigma_{\n} \mu \sqrt{r}}
	\end{align*}
	
	In the above, we used
	(i) $\sum_k (\z^\top \b_k)^2 = \|\z^\top \B\|^2 \le \|\B\|^2$  since $\z$ is unit norm,
	(ii) Lemma \ref{Blemma_new1} item 7 to bound $ \|\B\| \le 1.1 \sigmax$, and
	(iii) Lemma \ref{Blemma_new1} item 2 to bound $\|\b_k\| \le 1.1 \mu \sigmax \sqrt{r/q}$. 
	
	For $\eps_2 < \frac{c}{\sqrt{r}\mu\kappa^2}$, $\frac{c}{\kappa^3}\leq \frac{\sqrt{q}\sigma_{\n}}{\sigmin}$ the first term above is smaller , i.e., $\min(\frac{t^2}{\sum_\ik K_\ik^2},\frac{t}{\max_\ik K_\ik} ) =  c \frac{\eps_2^2 m \sigmin^2 }{ \sigma_{\n}^2 \kappa^2 }.$
	Thus, by sub-exponential Bernstein,  w.p. at least  $1-\exp (- c \frac{\eps_2^2 m \sigmin^2 }{ \sigma_{\n}^2 \kappa^2 } )- 2\exp(\log q + r - c m) - \exp (\log q + r-\frac{m\eps^2\mu^2r\sigmax{}^2}{Cq\sigma_{\n}^2} )$, for a given $\w,\z$,
	\[
	\w^\top\mathrm{Term_{\mathrm{noise}}}\ \z \le  \eps_2 m  \sigmin^2
	\]
	Using a standard epsilon-net argument to bound the maximum of the above over all unit norm $\w,\z$, e.g., using \cite[Proposition 4.7]{lrpr_gdmin}, we can conclude that
	\[
	\|\mathrm{Term_{\mathrm{noise}}} \| \le  1.1 \eps_2 m  \sigmin^2
	\]
	w.p. at least $1-\exp (C(n+r)- c \frac{\eps_2^2 m q\sigmin^2 }{ \sigma_{\n}^2 q\kappa^2 } )- 2\exp(\log q + r - c m) - \exp (\log q + r-\frac{m\eps^2\mu^2r\sigmax{}^2}{Cq\sigma_{\n}^2} )$. The factor of $ \exp(C(n+r))$ is due to the epsilon-net over $\w$ and that over $\z$: $\w$ is an $n$-length unit norm vector while $\z$ is an $r$-length unit norm vector. The smallest epsilon net covering the hyper-sphere of all $\w$'s  is of size $(1+ 2/\eps_{net})^n = C^n$ with $\eps_{net}=c$ while that for $\z$ is of size $C^r$. Union bounding over both thus gives a factor of $C^{n+r}$.
	By replacing $\eps_2$ by $\eps_2/1.1$, our bound becomes simpler (and $1/1.1^2$ gets incorporated into the factor $c$).

	Combining bounds for $\|\mathrm{Term_{\mathrm{nonoise}}}\|$ and $\|\mathrm{Term_{\mathrm{noise}}}\|$ we get  w.p. at least $1-  \exp ( C(n+r) -c\frac{\eps_1^2 m q }{\kappa^4 \mu^2 r} )-\exp (C(n+r)- c \frac{\eps_2^2 m q\sigmin^2 }{ \sigma_{\n}^2 q\kappa^2 } )- 2\exp(\log q + r - c m) - \exp (\log q + r-\frac{m\eps^2\mu^2r\sigmax{}^2}{Cq\sigma_{\n}^2} )$
	
	\[\norm{\nabla f(\U,\B)-\E[\nabla f(\U,\B)]}\leq \eps_1 \delta_t m \sigmin^2+ \eps_2 m  \sigmin^2	\]
\end{proof}
\section{Proof of the Initialization result}\label{init_thms}
\begin{lemma} \label{Wedinlemma}
	Conditioned on $\alpha$, we have the following conclusions.
	\ben
	\item Let $\bzeta$  be a scalar standard Gaussian r.v.. Define
	\[
	\w_k(\alpha)=\E\left[\bzeta^2 \indic_{ \left\{ |\bzeta| \leq \frac{\sqrt{\alpha}}{\sqrt{\|\xstar_{k}\|^2+\sigma_{\n}^2}} \right\} } \right].
	\]
	Then,
	\begin{align}
		&\E[\Xhat_0|\alpha] = \Xstar \D(\alpha), \nonumber\\
		&\text{ where }  \D(\alpha):=diagonal(\w_k(\alpha),k \in [q])
		\label{X0}
	\end{align}
	i.e. $\D(\alpha)$ is a diagonal matrix of size $q\times q$ with diagonal entries $\w_k(\alpha)$ defined above.
	
	\item Let $\E[\Xhat_{0}|\alpha] = \Xstar \D(\alpha) \svdeq \Ustar \check\bSigma \Bcheck$ be its $r$-SVD. Then,
	\begin{align}\label{Wedin_main}
		& \SD_2(\U_0,\Ustar) \le \nonumber \\
		&   \dfrac{\sqrt{2} \max\left( \| (\Xhat_0 - \E[\Xhat_{0}|\alpha] )^\top \Ustar \|_F , \| (\Xhat_0 - \E[\Xhat_{0}|\alpha] )  \Bcheck{}^\top \|_F \right)}{\sigmin \min_k \w_k(\alpha) - \|\Xhat_0 - \E[\Xhat_{0}|\alpha] \|}
	\end{align}
	as long as the denominator is non-negative.
	\een
\end{lemma}
Define the set $\mathcal{E}$ as follows
\begin{align}\label{def_ev}
	&\mathcal{E}:= \left\{ \tC(1 - \epsilon_1)\left(\frac{\|\Xstar\|_F^2}{q}+ \sigma_{\n}^2\right)  \le \alpha \right. \\ &\left. \qquad \qquad \le \tC (1 + \epsilon_1) \left(\frac{\|\Xstar\|_F^2}{q}+ \sigma_{\n}^2\right)  \right\}.
\end{align}
\begin{fact}\label{sumyik_bnd}
	$\Pr(\alpha \in \mathcal{E}) \ge 1 - \exp(- \tc mq \epsilon_1^2):= 1 - p_\alpha$.  Here  $\tc = c/\tC = c / \kappa^2 \mu^2.$
\end{fact}
\begin{proof}
	\[\alpha=\tC\frac{\sum_{\ik}(\y_{\ik})^2}{mq}=\tC\frac{\sum_{\ik}(\a_{\ik}^\top\xstar_k + \n_{\ik})^2}{mq}\]	
	
	The summands are independent sub-exponential r.v.s with sub-exponential norm $C(\norm{\xstar_k}^2+\sigma_{\n}^2)$. The fact is then an immediate consequence of sub-exponential Bernstein inequality for bounding $\left|\alpha - \left(\frac{\|\Xstar\|_F^2}{q}+ \sigma_{\n}^2\right)\right|$.
\end{proof}
\begin{fact}
	\label{betak_bnd}
	For any $\epsilon_1 \leq 0.1$,
	$
	\min_k  \w_k(\alpha)=\min_k\E\left[\bzeta^2 \indic_{ \left\{ |\bzeta| \leq \frac{\sqrt{\tC(1-\epsilon_1)}}{\sqrt{\|\xstar_{k}\|^2+\sigma_{\n}^2}}\sqrt{\frac{\|\Xstar\|_F^2}{q}+\sigma_{\n}^2} \right\} } \right] \geq 0.92.
	$
\end{fact}
\begin{proof}
	Let $\gamma_k = \frac{\sqrt{\tC(1-\epsilon_1)\left(\frac{\|\Xstar\|_F^2}{q}+\sigma_{\n}^2\right)}}{\sqrt{\|\xstar_{k}\|^2+\sigma_{\n}^2}}$. 
	Since $\|\xstar_{k}\|^2\leq \mu^2\kappa^2\|\Xstar\|_F^2/q$ (Assumption \ref{right_incoh}) we get
	\begin{align*}
		\gamma_k&\geq \frac{\sqrt{\tC(1-\epsilon_1)\left(\frac{\|\Xstar\|_F^2}{q}+\sigma_{\n}^2\right)}}{\sqrt{\mu^2\kappa^2\frac{\|\Xstar\|_F^2}{q}+\sigma_{\n}^2}}\\
		& \text{Since $\mu^2\kappa^2\geq 1$} \\
		&\geq \frac{\sqrt{\tC(1-\epsilon_1)\left(\frac{\|\Xstar\|_F^2}{q}+\sigma_{\n}^2\right)}}{\mu\kappa\sqrt{\frac{\|\Xstar\|_F^2}{q}+\sigma_{\n}^2}}\\
		&\text{Since $\tC = 9\mu^2\kappa^2$}\\
		&\geq 3
	\end{align*}
	Now,
	\begin{align*}
		\E\left[ \bzeta^2\indic_{\left\{ |\bzeta| \leq \gamma_k \right\}} \right]
		=& 1 - \E\left[ \bzeta^2\indic_{\left\{ |\bzeta| \geq \gamma_k \right\}} \right]\\
		\ge &  1 - \frac{2}{\sqrt{2\pi}}\int_{3}^{\infty}z^2\exp(-z^2/2)dz \\
		\geq & 1 - \frac{2e^{-1/2}}{\sqrt{\pi}}\int_{3}^{\infty}z\exp(-z^2/4)dz
		\\&=  1 - \frac{2e^{-11/4}}{\sqrt{\pi}}  \geq 0.92.
	\end{align*}
	The first inequality used $\gamma_k \ge 3$. The second used the fact that $z\exp(-z^2/4) \leq \sqrt{2e}$ for all $z\in \Re$.
\end{proof}

\begin{lemma} \label{init_terms_bnd}
	Fix $0 < \epsilon_1 < 1$ then,
	\ben
	\item \label{Xhat0_1}
	w.p. at least $1-\exp\left[(n+q)-c\epsilon_1^2mq/\mu^2\kappa^2\right]$, conditioned on $\alpha$, for an $\alpha \in\mathcal{E} $,
	\[
	\|\Xhat_{0} -\E[\Xhat_{0}|\alpha]\| \leq 1.1 \epsilon_1 \left(\|\Xstar\|_F+ \sqrt{q}\sigma_{\n}\right)
	\]
	
	\item
	\label{Xhat0_Ustar_1}
	w.p. at least $1-\exp\left[ qr - c \epsilon_1^2 mq / \mu^2\kappa^2\right]$, conditioned on $\alpha$, for an $\alpha \in \mathcal{E}$,
	\[
	\|\left(\Xhat_{0} - \E[\Xhat_{0}|\alpha]\right){}^\top\Ustar\|_F \leq 1.1 \epsilon_1 \left(\|\Xstar\|_F+ \sqrt{q}\sigma_{\n}\right)
	\]
	
	\item
	\label{lem:init_nom_B_term2}
	\label{Xhat0_Bstar_1}
	w.p. at least $1-\exp\left[nr - c \epsilon_1^2mq/\mu^2\kappa^2\right]$, conditioned on $\alpha$, for an $\alpha \in \mathcal{E}$,
	\[
	\|\left(\Xhat_{0} - \E[\Xhat_{0}|\alpha]\right)\Bcheck{}^\top\|_F \leq 1.1  \epsilon_1 \left(\|\Xstar\|_F+ \sqrt{q}\sigma_{\n}\right)
	\]
	
	\een
\end{lemma}
\begin{proof}[Proof of Theorem \ref{init_thm}]
	Set  $\epsilon_1 = 0.4 \delta_0 \sigmin/ \left(\|\Xstar\|_F+ \sqrt{q}\sigma_{\n}\right) $. Define
	$
	p_0 = 2\exp( (n+q)- c mq \delta_0^2\sigmin{}^2 /\mu^2\kappa^2(r\sigmax{}^2 + q\sigma_{\n}^2)  ) + 2\exp( n r - c mq \delta_0^2\sigmin{}^2 / \mu^2\kappa^2(r\sigmax{}^2 + q\sigma_{\n}^2) ) + 2\exp( q r - c mq \delta_0^2\sigmin{}^2 / \mu^2\kappa^2(r\sigmax{}^2 + q\sigma_{\n}^2) ).
	$
	Recall that  $\Pr(\alpha \in \mathcal{E}) \ge 1 - p_\alpha$ with
	$
	p_\alpha = \exp(- \tc mq \epsilon_1^2) = \exp(- c mq \delta_0^2\sigmin{}^2  /\mu^2\kappa^2(r\sigmax{}^2 + q\sigma_{\n}^2) ).
	$
	
	Using Lemma \ref{init_terms_bnd}, conditioned on $\alpha$, for an $\alpha \in \mathcal{E}$,
	\bi
	\item  w.p. at least $ 1-p_0$,
	$
	\|\Xhat_0 - \E[\Xhat_{0}|\alpha]\| \le 1.1  \epsilon_1 \left(\|\Xstar\|_F+ \sqrt{q}\sigma_{\n}\right)  =  0.44 \delta_0 \sigmin,
	$ and %
	{\small
		$
		\max\left( \| (\Xhat_0 - \E[\Xhat_{0}|\alpha])^\top \Ustar \|_F , \| (\Xhat_0 - \E[\Xhat_{0}|\alpha])  \Bcheck^\top \|_F \right) \le 0.44 \delta_0 \sigmin
		$
	}
	\item 
	$\min_k  \w_k(\alpha) \ge 0.9
	$
	
	The first inequality is an immediate consequence of $\alpha \in \mathcal{E}$ and the second follows by Fact \ref{betak_bnd}.
	\ei

	Plugging the above bounds into \eqref{Wedin_main} of Lemma \ref{Wedinlemma}, conditioned on $\alpha$, for any $\alpha \in \mathcal{E}$, w.p. at least $ 1 - p_0$,
	$\SEF(\U_0,\Ustar)  \le  \frac{0.44 \delta_0}{0.9 - 0.44 \delta_0} < \delta_0$ since $\delta_0 < 0.1$. In other words,
	\begin{align}
		& \Pr\left( \SEF(\U_0,\Ustar)  \ge  \delta_0 | \alpha \right) \le p_0 \ \text{for any $\alpha \in \mathcal{E}$}.
		\label{SD_given_alpha_bnd}
	\end{align}
	Since (i)
	$
	\Pr( \SEF(\U_0,\Ustar)  \ge \delta_0 )
	\le \Pr( \SEF(\U_0,\Ustar)  \ge  \delta_0 \text{ and }  \alpha \in \mathcal{E}) +  \Pr (\alpha \notin \mathcal{E}),
	$ and
	(ii)
	$
	\Pr( \SEF(\U_0,\Ustar)  \ge \delta_0 \text{ and } \alpha \in \mathcal{E} ) \le \Pr(\alpha \in \mathcal{E}) \max_{\alpha \in \mathcal{E}}\Pr( \SEF(\U_0,\Ustar)  \le  \delta_0 |\alpha ),
	$
	thus, using Fact \ref{sumyik_bnd} and \eqref{SD_given_alpha_bnd}, we can conclude that
	\[
	\Pr \left( \SEF(\U_0,\Ustar)  \ge \delta_0 \right) \le p_0 (1 - p_\alpha)  + p_\alpha \le p_0 + p_\alpha
	\]
	
	Thus, for a $\delta_0 < 0.1$, $\SEF(\U_0,\Ustar) < \delta_0$ w.p. at least $ 1- p_0 - p_\alpha = 1 - 2\exp( (n+q)- c mq \delta_0^2\sigmin{}^2 /\mu^2\kappa^2(r\sigmax{}^2 + q\sigma_{\n}^2)  ) - 2\exp( n r - c mq \delta_0^2\sigmin{}^2 / \mu^2\kappa^2(r\sigmax{}^2 + q\sigma_{\n}^2) ) - 2\exp( q r - c mq \delta_0^2\sigmin{}^2 / \mu^2\kappa^2(r\sigmax{}^2 + q\sigma_{\n}^2) ) - \exp(- c mq \delta_0^2\sigmin{}^2  /\mu^2\kappa^2(r\sigmax{}^2 + q\sigma_{\n}^2) ) $.
	This is $\ge 1 - 5 \exp(-c(n+q))$ if $mq > C \kappa^2 \mu^2 (n+q)(r\sigmax{}^2 + q\sigma_{\n}^2) / \delta_0^2\sigmin{}^2$. This finishes our proof.
\end{proof}

\end{document}